\documentclass[preprint,10pt]{elsarticle}
\makeatletter
\def\ps@pprintTitle{%
\let\@oddhead\@empty
\let\@evenhead\@empty
\def\@oddfoot{\centerline{\thepage}}%
\let\@evenfoot\@oddfoot}
\makeatother

\journal{Chaos, Solitons \& Fractals}

\usepackage{graphicx,bbm,setspace,amssymb,amsfonts,amsmath,mathtools,mathrsfs,stmaryrd,amsthm,bm,etoolbox,comment,hyperref,url,caption,subcaption,color,enumitem,algorithm,booktabs,microtype,nicefrac,lingmacros,nccmath,tree-dvips,tikz-cd}
\usepackage[noend]{algpseudocode}
\usepackage{algorithm}

\usepackage[colorinlistoftodos]{todonotes}
\usepackage[utf8]{inputenc} % allow utf-8 input
\usepackage[T1]{fontenc}    % use 8-bit T1 fonts

\usepackage[noend]{algpseudocode}

\newtheorem{theorem}{Theorem}[section]

\newtheorem{thm-defn}[theorem]{Theorem/Definition}

\newtheorem{proposition}[theorem]{Proposition}

\theoremstyle{definition}

\theoremstyle{remark}

\newcommand{\ignore}[1]{}{}

\begin{document}

\begin{frontmatter}

\title{Geometric persistence and distributional trends in worldwide terrorism}

\author[label1,label2]{Nick James} %\ead{nick.james@unimelb.edu.au} 
\author[label3]{Max Menzies} \ead{max.menzies@alumni.harvard.edu}
\author[label4]{James Chok}
\author[label5]{Aaron Milner} % 
\author[label6]{Cas Milner}
\address[label1]{School of Mathematics and Statistics, University of Melbourne, Victoria, Australia}
\address[label2]{Melbourne Centre for Data Science, University of Melbourne, Victoria, Australia}
\address[label3]{Beijing Institute of Mathematical Sciences and Applications, Tsinghua University, Beijing, China}
\address[label4]{School of Mathematics, The University of Edinburgh, Edinburgh, United Kingdom}
\address[label5]{Lyndon Baines Johnson School of Public Affairs, University of Texas at Austin, Texas, United States}
\address[label6]{Department of Physics, Southern Methodist University, Texas, United States}

\begin{abstract}
This paper introduces new methods for studying the prevalence of terrorism around the world and over time. Our analysis treats spatial prevalence of terrorism, the changing profile of groups carrying out the acts of terrorism, and trends in how many attacks take place over time. First, we use a time-evolving cluster analysis to show that the geographic distribution of regions of high terrorist activity remains relatively consistent over time. Secondly, we use new metrics, inspired by geometry and probability, to track changes in the distributions of which groups are performing the terrorism. We identify times at which this distribution changes significantly and countries where the time-varying breakdown is most and least homogeneous. We observe startling geographic patterns, with the greatest heterogeneity from Africa. Finally, we use a new implementation of distances between distributions to group countries according to their incidence profiles over time. This analysis can aid in highlighting structural similarities in outbreaks of extreme behavior and the most and least significant public policies in minimizing a country's terrorism.

\end{abstract}

\begin{keyword}
Terrorism \sep Cluster analysis \sep Time series analysis \sep Metric geometry
\end{keyword}

\end{frontmatter}

%---------------------------- BODY OF PAPER---------------
\section{Introduction}
\label{sec:intro}

%0. Intro to terrorism
Terrorism and other violent coordinated behavior have a highly deleterious effect on society. Terrorism not only may cause loss of life and injury, it may also thoroughly destabilize communities and governments \cite{UNterrorism}, cause enormous economic costs \cite{Nato_review}, affect financial markets \cite{Arif2017}, and damage key industries such as tourism \cite{Dory2021}.

%1. General overview of why studying terrorism matters
Understanding trends in terrorism over time assists governments in anticipating and responding to its acute impact. The observation of geographic spread may inspire new approaches to geopolitical engagement \cite{Africa_geopolitics}; a changing composition of active terrorist organizations may require changes in negotiation or military response \cite{Woodward2005}; an increasing number of terrorist activities is cause for concern and could warrant a variety of measured responses.

%2. The definition of terrorism and some remarks, we use this GTD
Our data comes from the Global Terrorism Database (GTD) \cite{GTD}, a comprehensive catalog of more than 200 000 incidents from 1970 through 2018.  Incident data include date, location, perpetrator (including group responsible), and several other features. We remark that the commonly used term ``terrorism'' is a complex concept with emotional, political, and legal dimensions. The organization maintaining the GTD defines terrorism as ``the threatened or actual use of illegal force and violence by a non-state actor to attain a political, economic, religious, or social goal through fear, coercion, or intimidation.'' The data collecting organization ``uses a series of inclusion criteria to systematically identify events for inclusion in the database.'' We do not opine on the definition of terrorism, only including events from this database without additional selection criteria.

%3. Time series analysis paragraph
This paper builds on a long literature of \emph{multivariate time series analysis}, developing a new mathematical method and a more extensive analysis of terrorism dynamics than previously performed. Existing methods of time series analysis include parametric models \cite{Hethcote2000} such as exponential  \cite{Chowell2016} or power-law models, \cite{Vazquez2006} and nonparametric methods such as distance analysis \cite{Moeckel1997,James2020_nsm}, distance correlation \cite{Szkely2007,Mendes2018,Mendes2019}, Bayesian approaches \cite{james2021_spectral} and network models \cite{Shang2020}. We draw upon numerous approaches that have been successfully applied in numerous disparate fields such as epidemiology \cite{james2021_CovidIndia,Li2021_Matjaz,Manchein2020,Blasius2020,james2021_TVO,Perc2020,Machado2020}, finance \cite{Drod2021_entropy,james_georg,Wtorek2021_entropy,james_arjun,Drod2001,james2022_stagflation,Gbarowski2019,james2021_MJW}, cryptocurrency \cite{Sigaki2019,Drod2020_entropy,James2021_crypto2,Drod2020,Wtorek2020}, crime \cite{james2022_guns,Perc2013}, fine art \cite{Perc2020_art,Sigaki2018_art} and other fields \cite{Ribeiro2012,james2021_hydrogen,Merritt2014,james2021_olympics,Clauset2015,james2022_CO2}. We hope our work can complement other mathematical analyses of terrorism, which are relatively small in number compared to other fields but cross dynamical systems \cite{Lopes2016}, network analysis \cite{Latora2004,Eiselt2018}, power models \cite{Clauset2007,Bohorquez2009}, Poisson models \cite{Allegrini2004}, agent-based models \cite{Elliott2004}, percolation models \cite{Galam2002,Galam2003,Galam2003_2} and other parametric models \cite{Helbing2014,Johnson2011}.

%4. Cluster analysis paragraph
\emph{Cluster analysis} is another common statistical method with successful applications to numerous fields. Designed to group data points according to similarity, clustering algorithms are highly varied - common examples are K-means \cite{Lloyd1982} and spectral clustering \cite{vonLuxburg2007}, which partition elements into discrete sets, and hierarchical clustering \cite{Ward1963,Szekely2005}, which does not specify a precise number of clusters. In this paper, we will use hierarchical clustering and K-means; the latter requires an initial choice of the number of clusters $k$. We draw upon several methods to address the subtle question of how to select this $k$. The primary application of cluster analysis in this paper is to reveal geographic persistence of terrorist activity over time, and to investigate relationships between different years regarding the composition of which perpetrator organizations are responsible for the totality of attacks.

%5. Structure of paper 
This paper is structured as follows. In Section \ref{sec:clusteranalysis}, we use cluster analysis to reveal geographic persistence, namely consistent patterns where terrorism is occurring over time. In Section \ref{sec:simplexdistances}, we analyze changes in the composition of perpetrator organizations, that is who is committing terrorism, both within specified countries and in observed clusters of Section \ref{sec:clusteranalysis}. We identify countries with the greatest diversity of terrorist actors and key times in which particular countries' perpetrator profile changed. In Section \ref{sec:TSdistributions}, we analyze the changing number of attacks over time, and analyze similarity between different countries' trajectories of attacks over time.

Together, these analyses provide insight about geographic persistence of terrorism over time (where it is occurring), the changing composition of terrorist organizations committing the attacks (who is committing the terrorism) and the change in the number of attacks over time (how much terrorism is occurring), all analyzed over time. This analysis may assist geopolitical and social researchers and policymakers working to identify the most and least successful countries in combating any increasing terrorist activity, and could provide blueprints for increased international collaboration and better relations between countries.

\section{Geographic persistence via clustering}
\label{sec:clusteranalysis}

In this section, we investigate where terrorism occurs, specifically, the geographic persistence of terrorism over time, grouping terrorist attacks by year. Our data spans 2000-2018 inclusive, a period of $T=19$ years. We begin with the year 2000 as the last two decades have been characterized by a considerable increase in the prevalence of and worldwide attention on terrorism. Indeed, in 1998, Osama Bin Laden first published al Qaeda's fatwa against the western world \cite{alqaedafatwa}, and 2000 marked the first wave of synchronized al Qaeda terrorist activity against the United States \cite{millenniumplots,USSCole}. Soon after, the 9/11 terrorist attacks occurred and the global war on terror began.

We aim to demonstrate a considerable persistence in the geographic distribution of terrorist activities across the world. For this purpose, we use cluster analysis to group events, as recorded in the GTD \cite{GTD}, by their geographic proximity on a year-by-year basis. By calculating various estimates of the total number of clusters, we show a considerable persistence in the geographic patterns of terrorist attacks each year.

Specifically, we index our years $t=1,...,T$. For each year, we consider the full collection of $N_t$ attacks observed in that year. Using the haversine formula \cite{haversine}, we calculate all pairwise geodesic distances between the locations of events for that year, generating a $N_t \times N_t$ matrix. We make several adjustments to the data set, such as excluding Australia, New Zealand and Japan from the data, each of which reports very few attacks, and separate the Western Hemisphere (the Americas) from the Eastern Hemisphere (Europe, Asia, Africa and Oceania). These two hemispheres are separated by large distances across oceans, and thus no geographic cluster of terrorist activity should cross over into both hemispheres. Separating them before implementing clustering is computationally efficient and avoids rare errors where isolated attacks in northern South America are connected to larger clusters in East Africa.

Finally, we draw upon $m=16$ different methods \cite{Charrad2014_nbclust} for selecting the number of clusters each year, to produce estimates $k_1(t),...,k_{m}(t)$ for each $t$. We record the estimated cluster numbers for each year of data and utilized the method in Table \ref{tab:americas} for the Western Hemisphere and Table \ref{tab:RoW} for the Eastern Hemisphere. We provide an overview of clustering and some of the methods used in  \ref{app:cluster}.

Almost every chosen method exhibits a considerable persistence in the detected number of clusters, suggesting a persistence in the geographic localization of terrorist attacks over time. By averaging over all chosen methods across all years, we estimate a persistent total of $K=4$ and 11 regions for the Western and Eastern hemispheres, respectively.

\begin{table*}
%\begin{sidewaystable}
%\begin{tabular}{lrrrrrrrrrrrrrrrrrrr}
\begin{tabular}{p{1.4cm}|p{0.7cm}|p{0.3cm}|p{0.3cm}|p{0.3cm}|p{0.3cm}|p{0.3cm}|p{0.3cm}|p{0.3cm}|p{0.3cm}|p{0.3cm}|p{0.3cm}|p{0.3cm}|p{0.3cm}|p{0.3cm}|p{0.3cm}|p{0.3cm}|p{0.3cm}|p{0.3cm}|p{0.3cm}}
\toprule
{} &  2000 & '01 & '02 & '03 & '04 & '05 & '06 & '07 & '08 & '09 & '10 & '11 & '12 & '13 & '14 & '15 & '16 & '17 & '18 \\
\midrule
ptbiserial &     4 &     5 &     5 &     4 &     3 &     5 &     5 &     4 &     4 &     3 &     5 &     5 &     5 &     4 &     4 &     5 &     5 &     4 &     3 \\
silhouette &     4 &     4 &     4 &     5 &     2 &     4 &     5 &     5 &     3 &     3 &     5 &     4 &     3 &     3 &     3 &     3 &     3 &     3 &     5 \\
kl         &     4 &     2 &     4 &     2 &     4 &     4 &     5 &     3 &     2 &     3 &     3 &     5 &     2 &     4 &     3 &     3 &     5 &     5 &     4 \\
cindex     &     2 &     2 &     5 &     5 &     3 &     5 &     5 &     3 &     3 &     5 &     4 &     3 &     5 &     2 &     4 &     4 &     5 &     5 &     4 \\
mcclain    &     4 &     3 &     2 &     3 &     2 &     5 &     2 &     2 &     4 &     2 &     2 &     2 &     5 &     2 &     2 &     2 &     2 &     2 &     2 \\
dunn       &     3 &     5 &     2 &     3 &     5 &     5 &     4 &     5 &     3 &     3 &     5 &     3 &     5 &     5 &     5 &     3 &     5 &     2 &     5 \\
ch         &     5 &     5 &     3 &     4 &     4 &     4 &     5 &     3 &     5 &     4 &     5 &     5 &     3 &     4 &     4 &     4 &     5 &     5 &     5 \\
hartigan   &     3 &     3 &     3 &     4 &     4 &     4 &     3 &     3 &     4 &     3 &     3 &     3 &     5 &     3 &     3 &     4 &     3 &     3 &     4 \\
ccc        &     2 &     2 &     3 &     2 &     2 &     5 &     2 &     5 &     2 &     2 &     5 &     5 &     2 &     5 &     5 &     5 &     5 &     5 &     5 \\
trcovw     &     4 &     3 &     3 &     3 &     3 &     3 &     3 &     3 &     3 &     3 &     3 &     3 &     3 &     3 &     3 &     3 &     3 &     3 &     3 \\
tracew     &     4 &     4 &     3 &     4 &     3 &     4 &     4 &     3 &     4 &     3 &     3 &     3 &     3 &     3 &     3 &     3 &     3 &     3 &     4 \\
friedman   &     4 &     5 &     5 &     4 &     5 &     4 &     4 &     4 &     4 &     4 &     5 &     5 &     3 &     4 &     5 &     4 &     5 &     4 &     4 \\
rubin      &     4 &     4 &     3 &     4 &     4 &     4 &     4 &     3 &     4 &     4 &     4 &     3 &     3 &     4 &     3 &     4 &     3 &     3 &     4 \\
db         &     5 &     4 &     4 &     5 &     3 &     5 &     2 &     4 &     3 &     5 &     5 &     4 &     5 &     4 &     4 &     3 &     4 &     5 &     3 \\
duda       &     2 &     5 &     3 &     2 &     2 &     3 &     2 &     3 &     3 &     3 &     3 &     3 &     3 &     3 &     3 &     3 &     3 &     3 &     0 \\
gap        &     2 &     2 &     3 &     2 &     2 &     2 &     2 &     3 &     2 &     2 &     5 &     5 &     2 &     3 &     3 &     3 &     3 &     3 &     2 \\
\bottomrule
\end{tabular}
\caption{Estimates of the number of clusters of terrorist attacks on a year-by-year basis across the Americas (Western Hemisphere). 16 different methods are compared (\cite{Charrad2014_nbclust}, Section 2 and Table 2), with considerable persistence over time in the number of clusters.}
\label{tab:americas}
%\end{sidewaystable}
\end{table*}

\begin{table*}
%\begin{tabular}{lrrrrrrrrrrrrrrrrrrr}
\begin{tabular}{p{1.4cm}|p{0.7cm}|p{0.3cm}|p{0.3cm}|p{0.3cm}|p{0.3cm}|p{0.3cm}|p{0.3cm}|p{0.3cm}|p{0.3cm}|p{0.3cm}|p{0.3cm}|p{0.3cm}|p{0.3cm}|p{0.3cm}|p{0.3cm}|p{0.3cm}|p{0.3cm}|p{0.3cm}|p{0.3cm}}

\toprule
{} &  2000 & '01 & '02 & '03 & '04 & '05 & '06 & '07 & '08 & '09 & '10 & '11 & '12 & '13 & '14 & '15 & '16 & '17 & '18  \\
\midrule
ptbiserial &     4 &     4 &     6 &     4 &     4 &     5 &     7 &     5 &     9 &     6 &     5 &     5 &     4 &     5 &     9 &     5 &     6 &     6 &     5 \\
silhouette &    25 &    25 &    25 &    25 &    15 &    14 &    17 &    13 &    17 &    20 &    13 &    16 &    10 &    14 &    24 &    23 &    14 &    14 &    23 \\
kl         &    17 &    11 &     7 &     7 &    15 &     7 &    20 &     8 &    10 &    24 &    13 &    16 &     7 &    14 &    23 &    18 &     9 &     9 &    13 \\
cindex     &    25 &    25 &    25 &    25 &    25 &    25 &    25 &    25 &    25 &    25 &    25 &    25 &    25 &    21 &    25 &    25 &    25 &    25 &    25 \\
mcclain    &     2 &     2 &     2 &     2 &     2 &     2 &     2 &     2 &     2 &     2 &     2 &     2 &     2 &     2 &     2 &     2 &     2 &     2 &     2 \\
dunn       &     8 &    25 &    12 &    22 &    25 &    21 &    22 &    24 &    21 &    11 &    21 &    15 &    20 &    13 &    25 &    16 &     8 &    11 &     2 \\
ch         &    25 &    25 &    24 &    25 &    17 &    17 &    18 &    14 &    17 &    20 &    13 &    17 &    10 &    14 &    13 &    17 &    14 &    14 &    13 \\
hartigan   &     8 &     7 &     6 &     4 &     6 &     7 &     9 &     8 &     3 &     5 &     6 &     7 &     6 &     5 &     8 &     5 &     6 &     8 &     7 \\
ccc        &    25 &    25 &    24 &    25 &    17 &    17 &    18 &    14 &    17 &    20 &    13 &    17 &    10 &    14 &    13 &    17 &    14 &    14 &    13 \\
trcovw     &     6 &     3 &     4 &     4 &     3 &     3 &     4 &     3 &     4 &     5 &     3 &     7 &     4 &     6 &     8 &     6 &     7 &     9 &     3 \\
tracew     &     4 &     7 &     4 &     4 &     3 &     3 &     4 &     8 &     4 &     5 &     6 &     7 &     4 &     6 &     8 &     6 &     7 &     4 &     4 \\
friedman   &    17 &    24 &    23 &    17 &    15 &    14 &    18 &    13 &    16 &    20 &    13 &    16 &    10 &    14 &    13 &    12 &    14 &    14 &    13 \\
rubin      &    17 &    24 &    24 &     7 &    15 &     7 &    16 &     8 &    16 &    16 &    13 &    16 &    10 &    14 &    23 &    17 &    12 &    14 &    13 \\
db         &    25 &    24 &    25 &     9 &    24 &    25 &    19 &    15 &     2 &    22 &     9 &    13 &    11 &     5 &     2 &     2 &    20 &     2 &    11 \\
duda       &     2 &     2 &     2 &     2 &     4 &     3 &     2 &     3 &     2 &     2 &     2 &     2 &     2 &     2 &     2 &     2 &     2 &     2 &     2 \\
gap        &     2 &     2 &     2 &     2 &     2 &     2 &     2 &     2 &     2 &     2 &     2 &     2 &     2 &     2 &     2 &     2 &     2 &     2 &     2 \\
\bottomrule
\end{tabular}
\caption{Estimates of the number of clusters of terrorist attacks on a year-by-year basis across the Eastern Hemisphere, comprising Europe, Africa and Asia, but excluding Australia, New Zealand and Japan. 16 different methods are compared (\cite{Charrad2014_nbclust}, Section 2 and Table 2), with considerable persistence over time in the number of clusters.}
\label{tab:RoW}
\end{table*}

\begin{table*}
\centering
\begin{tabular}{ccccccc}
\toprule
{} & \multicolumn{2}{c}{Americas} & \multicolumn{2}{c}{Eastern Hemisphere} & \multicolumn{2}{c}{Worldwide}\\
\hline
{} & Rand & Adjusted & Rand & Adjusted  & Rand & Adjusted \\
\midrule
2000 &                0.931 &                    0.857 &                   0.907 &                       0.692 &                 0.927 &                     0.716 \\
2001 &                0.934 &                    0.852 &                   0.974 &                       0.921 &                 0.980 &                     0.926 \\
2002 &                0.990 &                    0.975 &                   0.968 &                       0.907 &                 0.976 &                     0.919 \\
2003 &                0.952 &                    0.897 &                   0.950 &                       0.860 &                 0.961 &                     0.872 \\
2004 &                0.857 &                    0.672 &                   0.947 &                       0.873 &                 0.952 &                     0.878 \\
2005 &                0.950 &                    0.899 &                   0.960 &                       0.902 &                 0.963 &                     0.905 \\
2006 &                0.683 &                    0.351 &                   0.923 &                       0.812 &                 0.926 &                     0.816 \\
2007 &                0.865 &                    0.726 &                   0.966 &                       0.915 &                 0.967 &                     0.916 \\
2008 &                0.979 &                    0.954 &                   0.978 &                       0.946 &                 0.980 &                     0.947 \\
2009 &                0.921 &                    0.807 &                   0.999 &                       0.997 &                 0.999 &                     0.997 \\
2010 &                0.979 &                    0.953 &                   0.997 &                       0.993 &                 0.997 &                     0.993 \\
2011 &                0.955 &                    0.898 &                   0.993 &                       0.984 &                 0.993 &                     0.985 \\
2012 &                0.991 &                    0.982 &                   0.941 &                       0.850 &                 0.944 &                     0.852 \\
2013 &                0.931 &                    0.855 &                   0.971 &                       0.931 &                 0.972 &                     0.932 \\
2014 &                0.956 &                    0.906 &                   0.954 &                       0.887 &                 0.956 &                     0.889 \\
2015 &                0.992 &                    0.985 &                   0.962 &                       0.905 &                 0.964 &                     0.906 \\
2016 &                0.944 &                    0.885 &                   0.964 &                       0.909 &                 0.966 &                     0.910 \\
2017 &                0.954 &                    0.906 &                   0.995 &                       0.986 &                 0.995 &                     0.986 \\
2018 &                0.740 &                    0.494 &                   0.990 &                       0.973 &                 0.991 &                     0.972 \\
\bottomrule
\end{tabular}
\caption{Rand index and adjusted Rand index comparing the clusters on the full dataset across all time and the results of clustering on each individual year. The values close to 1 indicate high similarity, thus this table corroborates the considerable  persistence of geographic clusters of terrorism over time.}
\label{tab:Rand}
\end{table*}

Having observed the persistence in these cluster numbers on a yearly basis, we now wish to aggregate the clusters from each year so we can study the changing composition of these persistent regions over time. For this purpose, we apply the K-medoids algorithm over the entire set of attacks over 2000-2018, and cluster them into $K=4$ groups for the Americas (the Western Hemisphere) and $K=11$ for the Eastern Hemisphere. Two resulting clusters (one from each hemisphere) have fewer than 10 events from across the entire period, and are removed as outlier clusters. This yields 13 geographic clusters or regions of the world's terrorist attacks over time.

In the Western Hemisphere, the regions observed are North America, Colombia and Ecuador, and southern South America. In the Eastern Hemisphere, the regions observed are southern Africa, East Africa, West Africa, Europe, the Middle East, South Asia, East China, central Russia and Siberia, Southeast Asia, and the Philippines (and nearby Oceanic islands). In subsequent sections, we make use of these determined regions, where we study changes in the composition of terrorist attacks, both on a country-by-country basis and within our identified regions.

To examine the robustness of our clustering regions over time, we use an alternative measure of the consistency of the cluster results in different years. For each year $t$ we consider two different clusterings on the set $S_t$ of attacks observed in that year. The first is the result of K-medoids clustering ($K=13$) applied directly to $S_t$. The second is the result of taking our 13 clusters over all time and subsetting/restricting to year $t$. As this produces two different clusterings on the same finite set $S_t$, we may compare them using the Rand or adjusted Rand index \cite{Rand1971}. Essentially, this is a comparison of how much the operations ``cluster the data'' and ``restrict the data to one year'' commute. Similar analyses to examine the stability of clusterings in financial markets have been performed by \cite{Alves2020}. We record the Rand and adjusted Rand indices for the Americas (Western Hemisphere), Eastern Hemisphere, and entire world and each year in Table \ref{tab:Rand}.

A value close to 1 indicates a high similarity between two clusterings. As we can see, these indices are consistently high, with the Rand index for the Eastern Hemisphere or entire world never falling below 0.9 in all the years of analysis. Slightly less geographic persistence is observed for the Western Hemisphere (the Americas), particularly in 2006 and 2018. In these years, K-medoids clustering allocates a number of attacks in North America (specifically Central America) to the cluster of northern South America (Colombia and Ecuador). That is, applying K-medoids clustering to the attacks in these two years in isolation produces a slight blurring between the determined regions of North America and Colombia and Ecuador, as listed previously in this section. With this mild exception, we see very high geographic persistence worldwide and over time.

\section{Analysis of the composition of perpetrators}
\label{sec:simplexdistances}

In this section, we turn to an analysis of which actors and groups (including group designations) are committing the acts of terrorism within each region. Specifically, among the top 40 countries with the greatest total counts of terrorist events, we investigate the changes in the composition of terror perpetrators on a yearly basis. For each determined country, we construct a $T \times T$ distance matrix between different distributions of perpetrators by year. In addition, we perform the same analysis for the determined regions from Section \ref{sec:clusteranalysis}.

\subsection{Distance between compositional distributions}
\label{sec:simplexdistancesmethod}

Given a year $t$ and a given country or region (where at least one attack occurred), let $A_t$ be the sets of terror perpetrators of attacks in that year, as recorded in the GTD database. For any $x \in A_t$, let $p^{(t)}_x$ be the proportion of attacks attributable to the group $x$. Many attacks  are unknown in their attribution, thus, ``unknown'' is a permissible group. As every attack carries a unique primary attribution, $\sum_{x \in A_t} p^{(t)}_x = 1$. So we may understand $\mathbf{p}^{(t)}$ as a probability vector of length $|A_t|$.

Analogously, let $p^{(s)}_y$ be the proportion of attacks attributable to the group $y \in A_s$. Define a distance between the two distributions of terrorism by 
\begin{align}
\label{eq:simplexdist}
    d(s,t)= \frac12 \sum_{z \in A_s \cup A_t} |p^{(s)}_z - p^{(t)}_z|.
\end{align}
This has the property that $d(s,t)=0$ if and only if years $s,t$ have an identical proportion of terrorist attacks among perpetrator groups. In addition, $d(s,t)\leq 1$, with equality if and only if $A_s$ and $A_t$ are disjoint, with no intersection in perpetrators between the two years. If $A_s$ is the empty set, we exclude the year $s$ from consideration. Thus, a single country or region $R$ produces a $S \times S$ matrix $D^R$, where $S$ is the number of years where at least one attack was recorded; for most regions and countries, $S=T$.

There are several enlightening interpretations of the distance defined in (\ref{eq:simplexdist}). As mentioned, the vector of proportions for any given year is a probability vector $\mathbf{p}^{(t)}$ of dimension $|A_t|$. As the sum of the elements is 1, $\mathbf{p}^{(t)}$ naturally lies on a simplex $\Delta=\{ (x_1,...,x_N): \sum x_i=1, x_i \geq 0 \}$ of dimension $N=|A_t|$. When we compute the distance (\ref{eq:simplexdist}), we are implicitly embedding both $\mathbf{p}^{(t)}$ and $\mathbf{p}^{(s)}$ in a larger simplex of dimension $N=|A_t \cup A_s|$. In fact, when we compute all distances $d(s,t)$ one may consider the vectors $\mathbf{p}^{(t)}$ as having been embedded in a larger simplex of dimension $N=|A|$ where $A = \cup_{t=1}^T A_t$ is the set of all recorded groups in the region across all time.

Now simplices (including the higher-dimensional embedded simplex) are convex sets, that is, any line segment between two points in a simplex lies entirely within the simplex. Thus, we were initially motivated to use an alternative distance to (\ref{eq:simplexdist}) that bears explaining. As each probability vector $\mathbf{p}^{(t)}$ lies within a larger simplex and the line segment between $\mathbf{p}^{(s)},\mathbf{p}^{(t)}$ lies entirely within the simplex, an initial choice of distance could be the geometric straight line distance along the line segment, that is,
\begin{align}
\label{eq:simplexdist_alt}
    \left(\sum_{z \in A_s \cup A_t} |p^{(s)}_z - p^{(t)}_z|^2\right)^\frac12.
\end{align}
This distance may have other uses in future works, as it can be interpreted as the straight line distance along a path of deformation (in the geometric sense) between $\mathbf{p}^{(s)}$ and $\mathbf{p}^{(t)}$. But we think it is necessary to justify why we use (\ref{eq:simplexdist}) over (\ref{eq:simplexdist_alt}). We present both a geometric and probabilistic explanation.

First, the simplex may be interpreted as (the positive quadrant of) a unit sphere in the vector space $\mathbb{R}^N$ under the $L^1$ norm $\| x\|=|x_1|+...+|x_N|$ and then imbuing $\mathbb{R}^N$ with this norm naturally motivates using it to measure the distance between $\mathbf{p}^{(s)}$ and $\mathbf{p}^{(t)}$, as we have done in (\ref{eq:simplexdist}) up to a constant factor. And secondly, the quantity in (\ref{eq:simplexdist}) can be reinterpreted precisely (including the factor of $\frac12$) probabilistically as the \emph{discrete Wasserstein distance} between $\mathbf{p}^{(s)}$ and $\mathbf{p}^{(t)}$. Essentially, let the set $A$ of all perpetrator groups in a region be endowed with a metric such that $d(x,y)=1$ if $x\neq y$ and 0 otherwise. Then the Wasserstein metric between two distributions on $A$, which measures the ``work'' (in the sense of physics) to change one distribution into the other, coincides with the formula (\ref{eq:simplexdist}). This is discussed in greater detail in \cite{James2021_geodesicWasserstein} and  \ref{app:discreteWass}.

Finally, we define $\nu(R)$, a measure of the total heterogeneity of the composition of attacks across our period of analysis. Given an $S \times S$ matrix $A$, we define its Frobenius norm by $\|A\|=\left( \sum_{i,j=1}^S A_{ij}^2  \right)^\frac12.$ Let $\nu(R)=\|D^R\|$, intuitively, regions with a low $\nu(R)$ experience attacks rather consistently from the same set of organizations year after year; those with a high $\nu(R)$ experience attacks from organizations that change their composition substantially with time.

We remark that we could repeat the above analysis by removing all attacks with ``unknown'' perpetrator prior to any computations of $p_x$ and $\nu(R)$. However, an ``unknown'' perpetrator should not merely be considered as randomly missing data. Terrorist groups regularly take credit for their attacks for a complex variety of reasons and ``credit-taking has the potential to tell observers a great deal about the nature of the threats groups pose'' \cite{terroristcredit}. Thus, significantly different proportions of known vs unknown attacks between countries or regions should be taken into account in any analysis of compositional differences between perpetrators.

In addition, we remark that the value of $\nu(R)$ may reduce considerably with smaller values of $S$, as a smaller $S$ results in far fewer terms in the sum $\|A\|=\left( \sum_{i,j=1}^S A_{ij}^2  \right)^\frac12.$ We consider this appropriate: countries or regions that have entire years with no terrorist activity should be considered to display less heterogeneity in the composition of attacks. Simply put, with many years of no activity, less terrorism is occurring in absolute terms, so we should consider less compositional changes as occurring.

\begin{table}
\centering
\begin{tabular}{{c}{c}{c}}
\toprule
Rank & Country & Norm\\
\midrule
1 & South Sudan & 9.17 \\
2 & Ukraine & 21.71 \\
3 & Pakistan & 28.4 \\
4 & Syria & 30.59 \\
5 & Afghanistan & 31.66 \\
6 & Libya & 44.54 \\
7 & Thailand & 48.55 \\
8 & Lebanon & 50.46 \\
9 & Iraq & 51.05 \\
10 & The Philippines & 52.68 \\
11 & Egypt & 57.14 \\
12 & Bangladesh & 57.25 \\
13 & Russia & 63.93 \\
14 & Colombia & 66.91 \\
15 & United Kingdom & 68.02 \\
16 & Mali & 72.76 \\
17 & Turkey & 73.36 \\
18 & Cameroon & 73.83 \\
19 & India & 91.39 \\
20 & Indonesia & 94.43 \\
21 & Greece & 94.73 \\
22 & Saudi Arabia & 96.02 \\
23 & West Bank and Gaza Strip & 105.18 \\
24 & Yemen & 108.38 \\
25 & Sudan & 112.23 \\
26 & Central African Republic & 113.80 \\
27 & Myanmar & 116.98 \\
28 & Israel & 117.52 \\
29 & Nigeria & 125.06 \\
30 & Algeria & 126.61 \\
31 & Sri Lanka & 127.01 \\
32 & Burundi & 129.00 \\
33 & France & 131.04 \\
34 & Germany & 131.84 \\
35 & Kenya & 150.94 \\
36 & United States & 157.05 \\
37 & Somalia & 160.44 \\
38 & Nepal & 169.58 \\
39 & Spain & 175.93 \\
40 & Democratic Republic of the Congo & 188.99 \\
\bottomrule
\end{tabular}
\caption{Ranking of the 40 countries with the greatest total counts of terrorist events according to their value of $\nu(R)$, defined in Section \ref{sec:simplexdistancesmethod}. Countries with smaller values display more homogeneity in the composition of terrorist attacks from year to year.}
\label{tab:countrynormranking}
\end{table}

\begin{table}
\centering
\begin{tabular}{{c}{c}{c}}
\toprule
Rank & Cluster & Norm\\
\midrule
1 & East China &   22.16\\
2 & Central Russia and Siberia &   32.88 \\
3 & South Asia & 45.96 \\
4 & The Philippines and proximity & 48.10\\
5 & The Middle East & 60.88\\
6 & Colombia and Ecuador & 64.97\\
7 & Southeast Asia & 75.11\\
8 & Europe &  89.61\\
9 & Southern Africa &  113.21\\
10 & East Africa &  114.43\\
11 & North America & 114.92\\
12 & Southern South America & 117.95\\
13 & West Africa & 129.90\\
\bottomrule
\end{tabular}
\caption{Ranking of the $K=13$ determined clusters/regions with the greatest total counts of terrorist events according to their value of $\nu(R)$, defined in Section \ref{sec:simplexdistancesmethod}. Regions with smaller values display more homogeneity in the composition of terrorist attacks from year to year.}
\label{tab:clusternormranking}
\end{table}

\subsection{Results}

In Table \ref{tab:countrynormranking}, we rank the 40 countries with the greatest total counts of terrorist events according to their value of $\nu(R)$, while Table \ref{tab:clusternormranking} ranks the 13 determined clusters of attacks from Section \ref{sec:clusteranalysis}. In Figure \ref{fig:compositionbarplots}, we present bar plots of changing perpetrator composition over time for various countries $R$, while Figure \ref{fig:simplexplots} plots hierarchical clustering on $D^R$ for the same countries.

The countries with the lowest norm are South Sudan, Ukraine, Pakistan, Syria and Afghanistan. While the first four countries are characterized by a large number of unknown attacks, Afghanistan, however, is dominated by attacks from the Taliban after 2003 (Figure \ref{fig:Afghanistan_bar}). The difference in composition between 2000-2002 and subsequent years can be seen in Figure \ref{fig:Afghanistan}. The persistence of such a concentration in the perpetration of terrorism indicates how prolific they were as a terror organization since the turn of the century \cite{Taliban}. The countries with the highest norm are the United States, Somalia, Nepal, Spain, and the Democratic Republic of the Congo. The United States (Figures \ref{fig:USA_bar} and \ref{fig:USA}) has attacks spread over many groups rather than any consistency between actors, while other countries, particularly Somalia, exhibit a regime change in which different groups claim the majority of attacks in different years. Specifically, from 2009-2018 inclusive, al-Shabaab \cite{alshabaab} is consistently the top perpetrator of terrorism in Somalia, accounting for over 80\% of attacks from 2013 onwards (Figures \ref{fig:Somalia_bar} and \ref{fig:Somalia}). From 2000-2008, most attacks are characterized as unknown, with al-Shabaab responsible for at most 15\% of attacks in any single year.

There are several other countries that exhibit noteworthy trends in the composition of perpetrator organizations. In Sri Lanka, ranked 31st in $\nu(R)$, the Liberation Tigers of Tamil Eelam, commonly known as the Tamil Tigers \cite{Nieto2008}, constituted at least 69\% of attack attributions from 2000-2009 inclusive, but that organization ceased to exist following the end of the Sri Lankan civil war \cite{SriLankanwar}. This considerable shift in composition can be seen in Figures \ref{fig:SriLanka_bar} and \ref{fig:SriLanka}. In Colombia, the Revolutionary Armed Forces of Colombia or FARC dominated terrorism between 2000-2015 inclusive, constituting the top perpetrator of terrorism for all but one of these years. After a historic peace deal \cite{Gluecker2021}, FARC combatants were reintegrated into society and attacks ceased. Beyond 2016, the National Liberation Army of Colombia (ELN) \cite{Blanco2020} was responsible for the majority of terrorist attacks (Figures \ref{fig:Colombia_bar} and \ref{fig:Colombia}). Finally, in the Philippines, the New People's Army \cite{harmon2020philippines} accounts for a majority of terrorism of known groups from 2000-2018. However, the recorded composition of perpetrators changes significantly after 2005, as the modal number of attacks are ``unknown'' every year from 2006 (Figures \ref{fig:Philippines_bar} and \ref{fig:Philippines}).

Finally, regarding the determined regions from Section \ref{sec:clusteranalysis}, we see startling geographic trends in the ordering of composition heterogeneity in Table \ref{tab:clusternormranking}. The five clusters/regions with the least $\nu(R)$ (indicating greater heterogeneity in perpetrator attribution) are all in Asia, while the three African regions all fall within the five clusters with the greatest value of $\nu(R)$.

\section{Time series analysis of attack counts}
\label{sec:TSdistributions}
In this final section, we analyze the prevalence of terrorism over time by country, specifically, how many attacks occur throughout the analysis window. 

\subsection{Distances between time series}

Consider a fixed country or region $R$. To examine the trajectories of attack counts on a finer basis, we divide our 19-year period of analysis into individual months $t=1,2,...,P$, where $P=12 T = 228.$ Let $z_R(t)$ be the number of attacks, regardless of attribution, observed in month $t$. To the time series $z(t), t=1,...,P$, we associate the following probability distribution:
\begin{align}
\label{eq:associateddistribution}
    f_R=\frac{1}{\sum_{s=1}^P z_R(s) } \sum_{t=1}^P z_R(t) \delta_t,
\end{align}
where $\delta_t$ is a Dirac delta distribution at $t$. That is, $f_R$ is a distribution that apportions to month $t$ the weight of the attacks observed in that month as a proportion of the total attacks across the whole period. We can then compare different trajectories as distributions by using the $L^1$-Wasserstein metric \cite{DelBarrio}: $D(R_1, R_2)=W_1 (f_{R_1}, f_{R_2}).$ We remark that this is the more traditional Wasserstein distance, between distributions on the real number line with its standard metric, than the discrete Wasserstein distance discussed in Section \ref{sec:simplexdistances}.

This distance has several advantageous properties over more commonly used discrepancy measures between normalized trajectories. For example, previous work \cite{jamescovideu} has used the $L^1$ norm and metric between normalized trajectories, defined as follows:
\begin{align}
\|z_i\|_1= \sum_{t=1}^T z_i(t) \\
\mathbf{v}_i=\frac{z_i}{\|z_i \|_1}\\
  d_{ij}=  \|\mathbf{v}_i - \mathbf{v}_j\|_1
\end{align}
This treats each time series $z_i(t)$ as a vector in $\mathbb{R}^P$, normalizes by its $L^1$ norm, and compares these normalized vectors with the $L^1$ metric \cite{Minkowski}. This distance is suitable in most instances but has some undesirable properties when measuring discrepancy between noisy time series. Specifically, this $L^1$ distance $d_{ij}$ has maximal possible value equal to 2 when $z_i(t)$ and $z_j(t)$ have disjoint support. Practically, this would mean that two countries' trajectories would receive a large $L^1$ discrepancy measure if the attacks fell around the same time but not in exactly the same months. For example, if country $i$ and $j$ had broadly similar trends in events, but in country $i$ more attacks were reported in January and March while country $j$ reported more on February and April, then the $L^1$ distance measure would be larger than their similarity. Further smoothing and averaging can resolve some of these issues, but the Wasserstein metric ameliorates this issue even more, as it is robust to small translations of distributions. That is, if $f$ is a distribution and $f_\delta(x)=f(x+\delta)$, then $W_1(f,f_\delta)=|\delta|$, as shown in \cite{james2021_portfolio}. This means the Wasserstein metric assigns a low value in the case that countries $i$ and $j$ have similar trajectories where attacks just fall in nearby but distinct months.

\subsection{Results}

We apply hierarchical clustering to this metric across the 40 countries with the greatest total counts of terrorist events (Figure \ref{fig:Wassersteincountries}) and the 13 regions identified in Section \ref{sec:clusteranalysis} (Figure \ref{fig:Wassersteinclusters}). Various insights can be gleaned from the cluster structure and pairings. First, the prototypical trajectory of terrorist events among our collection of countries is captured in a large cluster of similarity that extends from Nigeria down to India in Figure \ref{fig:Wassersteincountries}. Containing 26 countries, the primary feature of this cluster is a general increase in terrorist activity over time. A limited amount of heterogeneity exists within this majority cluster, visible in the slight differences between Nigeria (Figure \ref{fig:NigeriaTS}) and India (\ref{fig:IndiaTS}). However, essentially these 26 countries display time series that, considered as distributions, have most of the terrorist attacks occurring toward the end of the period of analysis. The increase in terrorist activity in Africa (with many of the time series resembling that of Nigeria) is particularly noticeable, and has been of growing concern in recent years \cite{Alvi2019,UN_Africa}.

Other than the majority behavior, some outlier countries are revealed. Sri Lanka, Algeria and Spain all share some unusual and striking commonalities in trends in terrorist attacks over time - a striking (relative) dearth of attacks after 2010. This is illustrated in Figures \ref{fig:SriLankaTS} and \ref{fig:SpainTS}, where the time series (considered as distributions) have the majority of their ``weight'' before 2010. Next, the United States and the West Bank and Gaza Strip are observed to be quite different from the majority collection. Indeed, as seen in Figures \ref{fig:USATS} and \ref{fig:GazaTS} respectively, these time series yield distributions that are relatively flat over the entire period of analysis, a rather rare feature among the countries studied. Other outlier countries include Indonesia and to a lesser extent Russia and Greece - all three of these exhibits somewhat idiosyncratic trajectories of terrorism over time, with bursts over differing periods and little relationship to the trends observed in other countries.

These geographic trends are also made apparent in Figure \ref{fig:Wassersteinclusters}. There, the regions where terrorist attacks mostly occur toward the end of the period group together in the bottom cluster, including regions in Africa, South Asia, and the Middle East. Europe and North America cluster together as being characterized more by terrorism, remaining flatter across the period of analysis.

\begin{figure*}
    \centering
    \begin{subfigure}[b]{0.49\textwidth}
        \includegraphics[width=\textwidth]{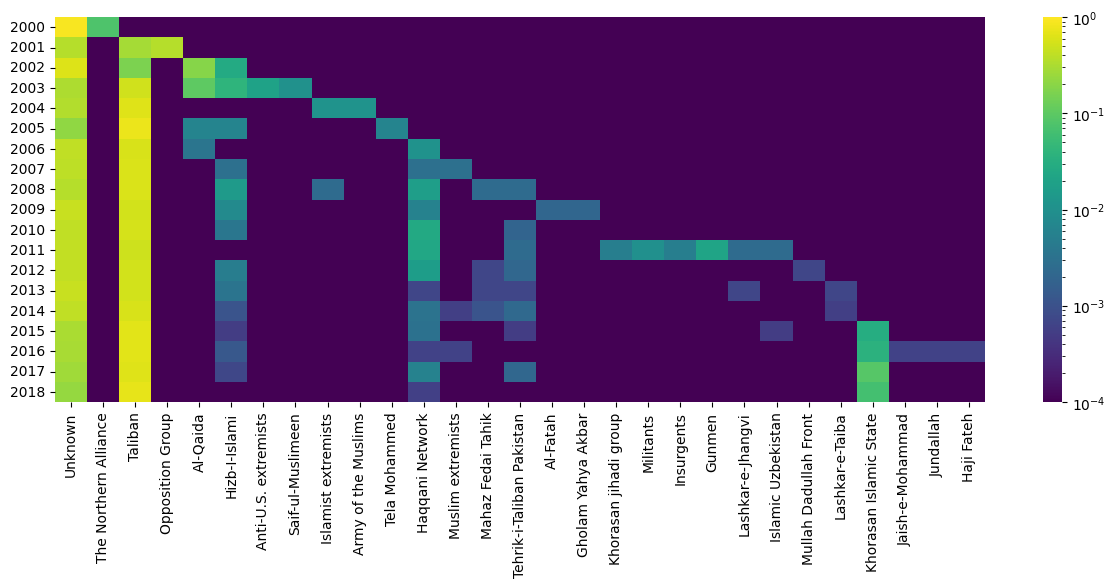}
        \caption{}
    \label{fig:Afghanistan_bar}
    \end{subfigure}
    \begin{subfigure}[b]{0.49\textwidth}
        \includegraphics[width=\textwidth]{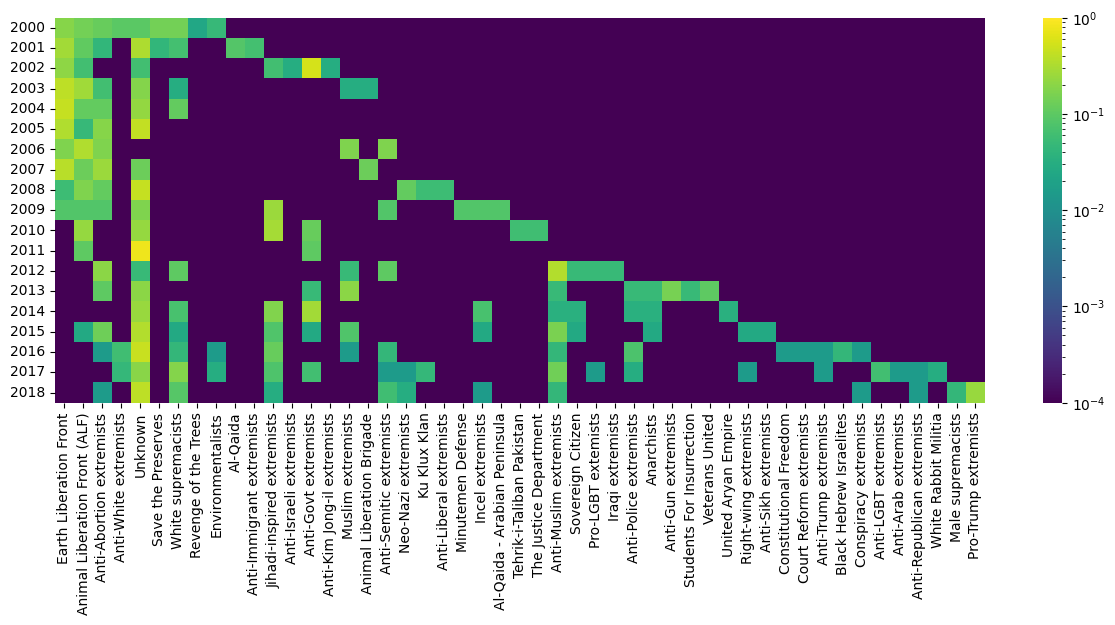}
        \caption{}
    \label{fig:USA_bar}
    \end{subfigure}
    \begin{subfigure}[b]{0.49\textwidth}
        \includegraphics[width=\textwidth]{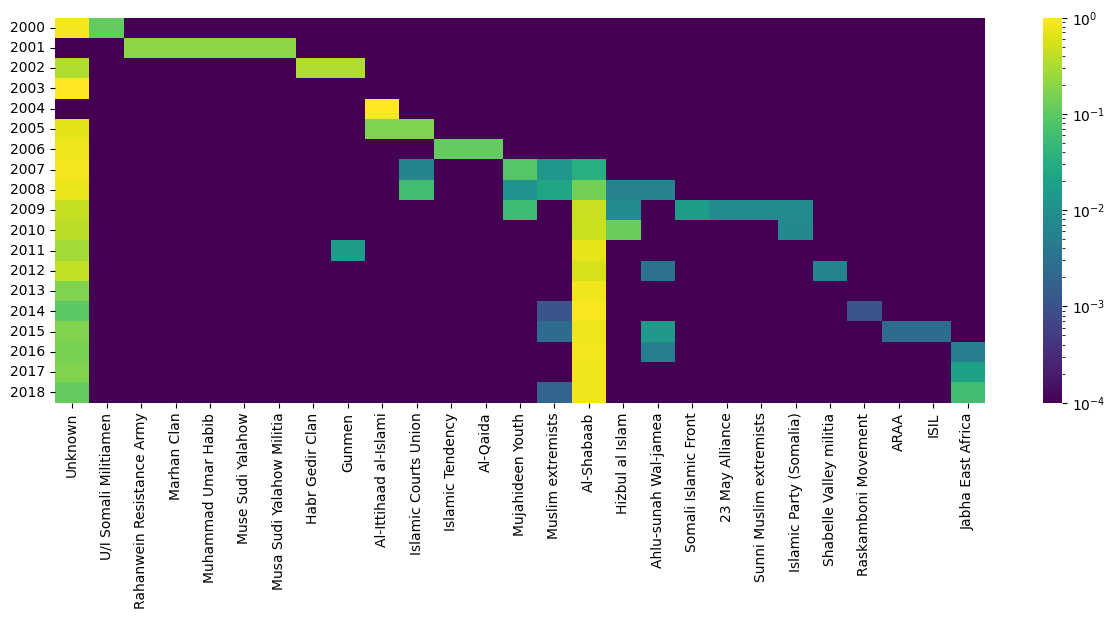}
        \caption{}
    \label{fig:Somalia_bar}
    \end{subfigure}
    \begin{subfigure}[b]{0.49\textwidth}
        \includegraphics[width=\textwidth]{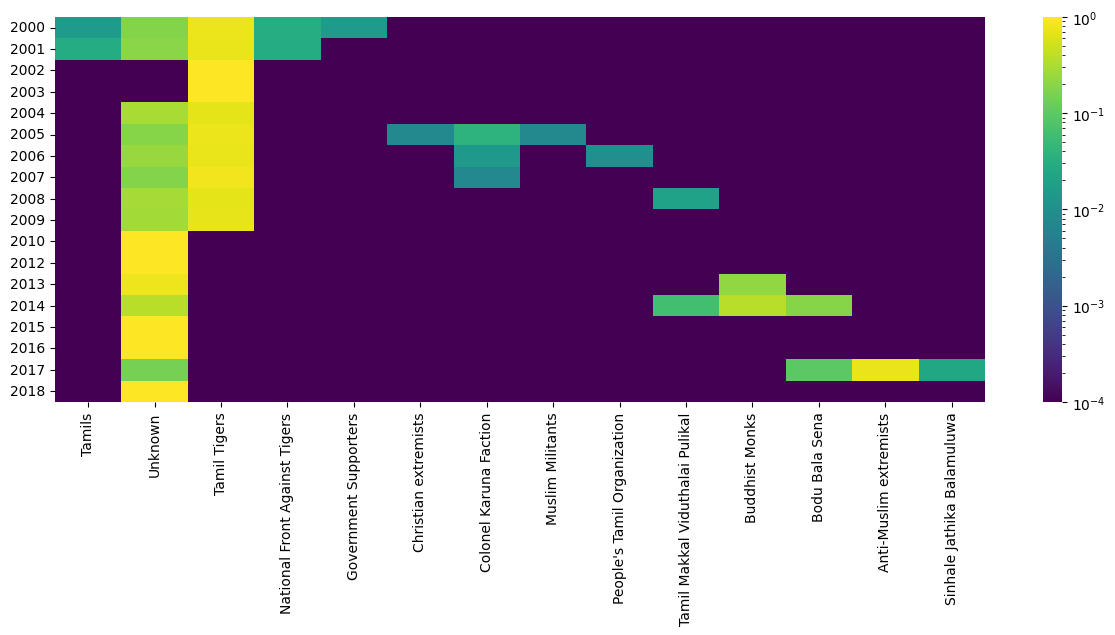}
        \caption{}
    \label{fig:SriLanka_bar}
    \end{subfigure}
    \begin{subfigure}[b]{0.49\textwidth}
        \includegraphics[width=\textwidth]{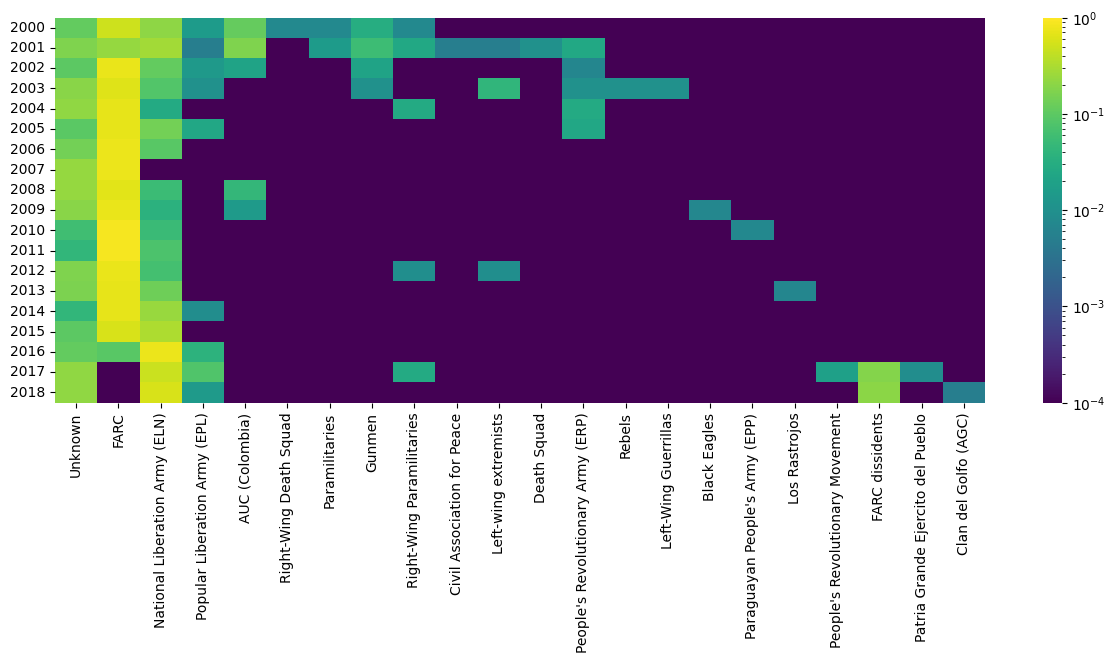}
        \caption{}
    \label{fig:Colombia_bar}
    \end{subfigure}
    \begin{subfigure}[b]{0.49\textwidth}
        \includegraphics[width=\textwidth]{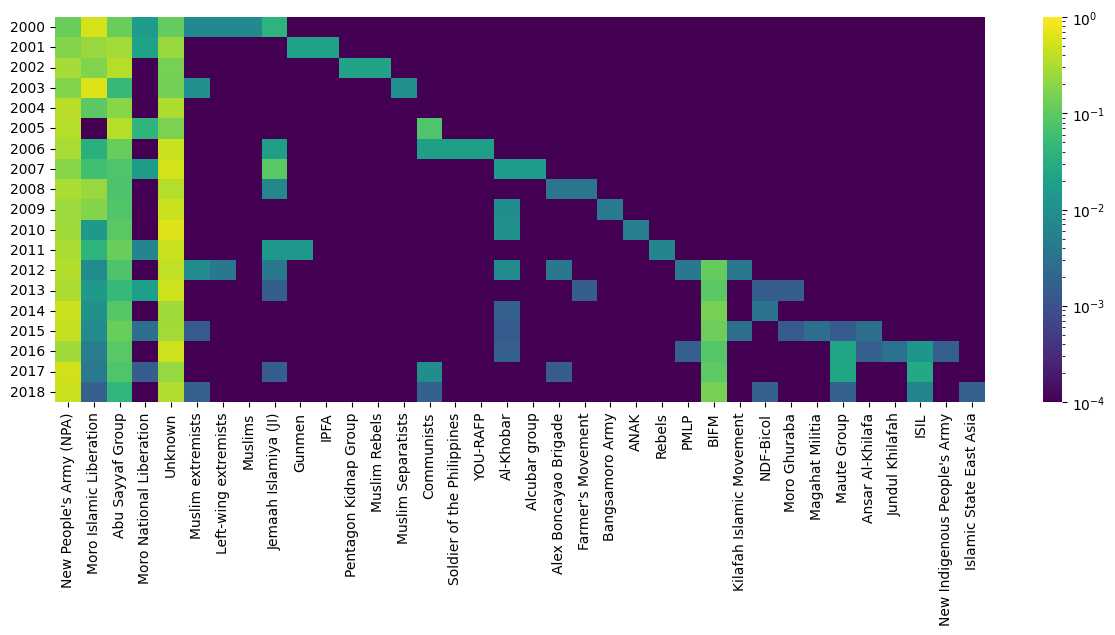}
        \caption{}
    \label{fig:Philippines_bar}
    \end{subfigure}
    \caption{Percentage of attacks of a given year attributed to a perpetrator for (a) Afghanistan, (b) the United States, (c) Somalia, (d) Sri Lanka, (e) Colombia, (f) the Philippines. We use a logarithmic scale for increased visibility of small percentages, adjusting 0 to $10^{-4}$.}
    \label{fig:compositionbarplots}
\end{figure*}

\begin{figure*}
    \centering
    \begin{subfigure}[b]{0.49\textwidth}
        \includegraphics[width=\textwidth]{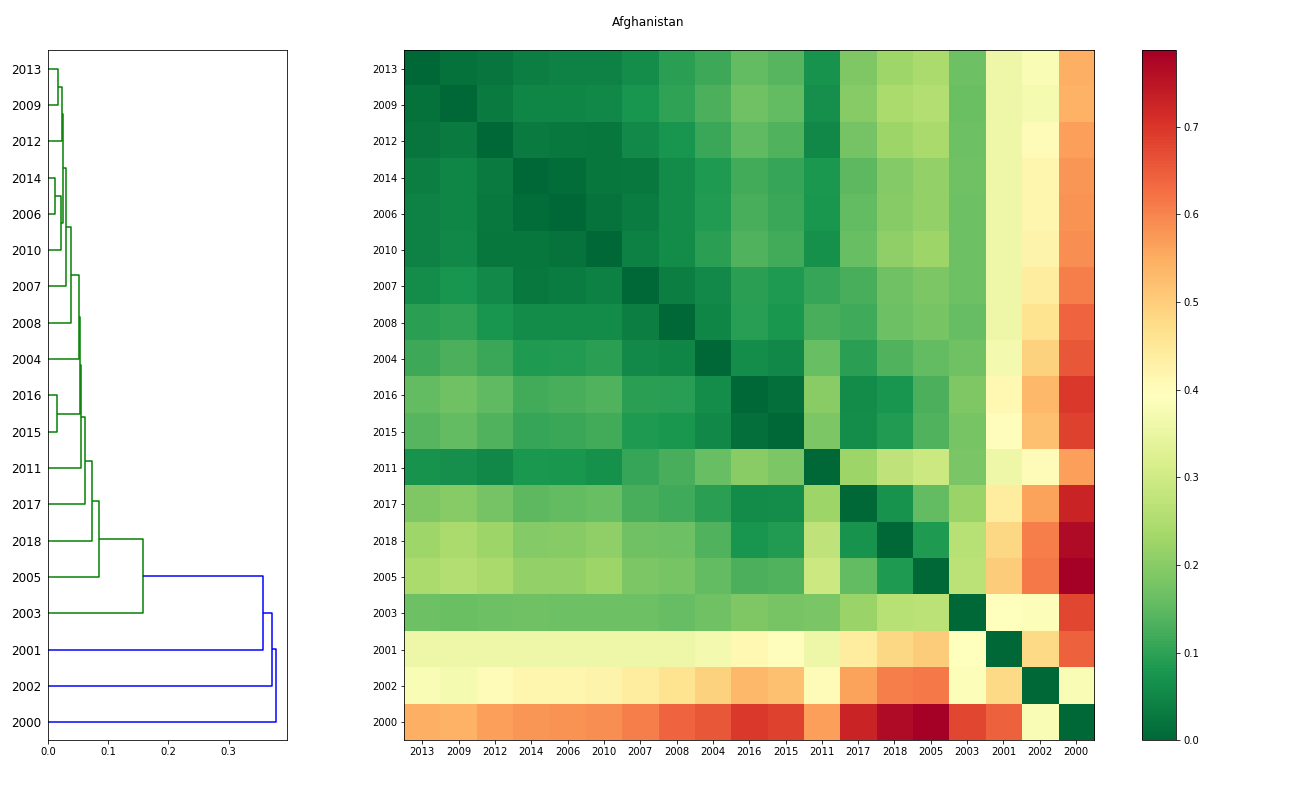}
        \caption{}
    \label{fig:Afghanistan}
    \end{subfigure}
    \begin{subfigure}[b]{0.49\textwidth}
        \includegraphics[width=\textwidth]{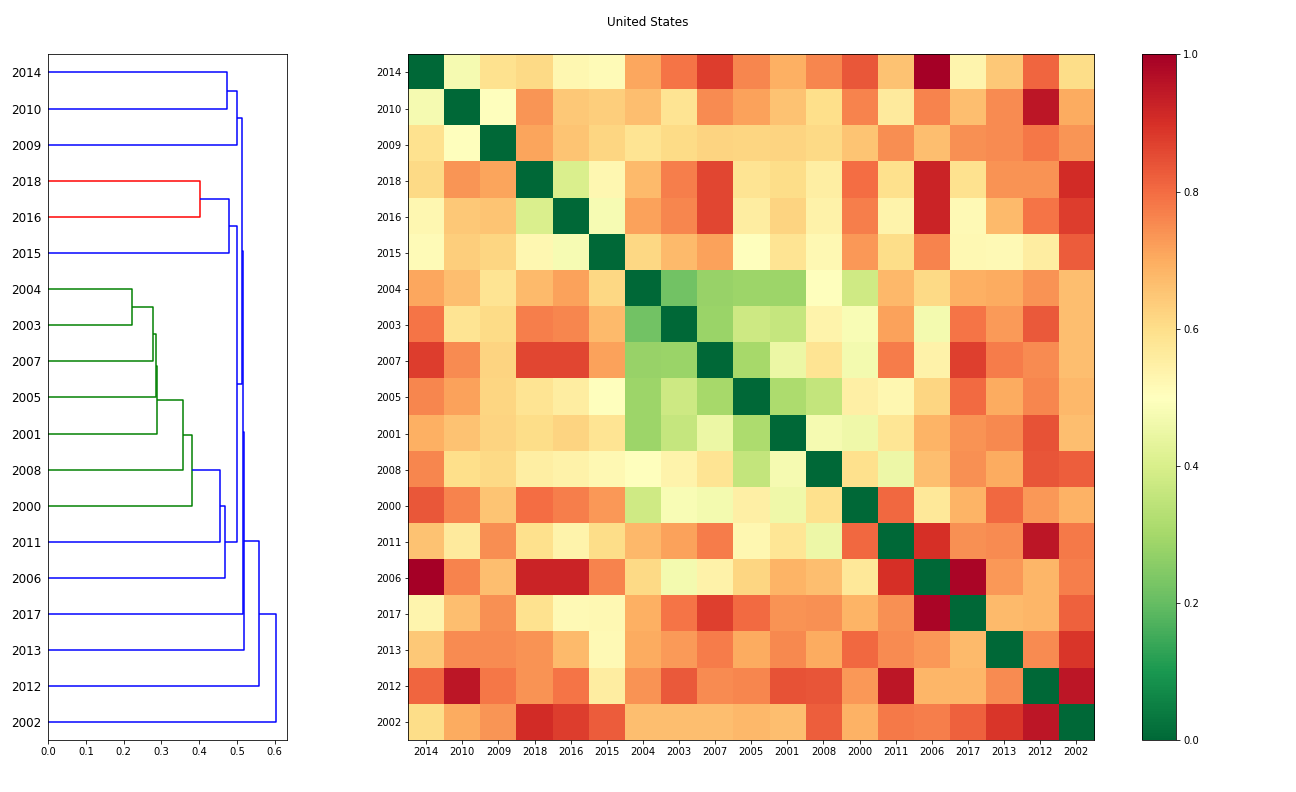}
        \caption{}
    \label{fig:USA}
    \end{subfigure}
    \begin{subfigure}[b]{0.49\textwidth}
        \includegraphics[width=\textwidth]{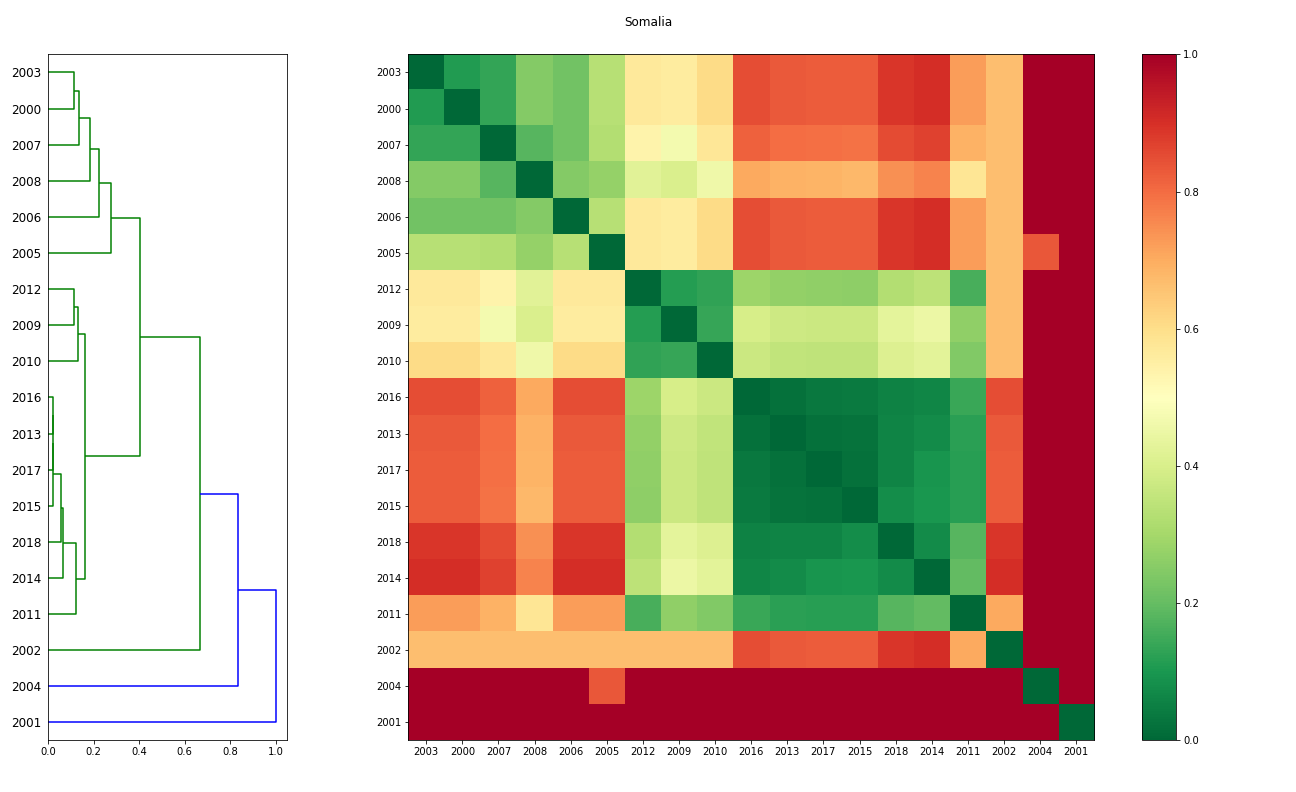}
        \caption{}
    \label{fig:Somalia}
    \end{subfigure}
    \begin{subfigure}[b]{0.49\textwidth}
        \includegraphics[width=\textwidth]{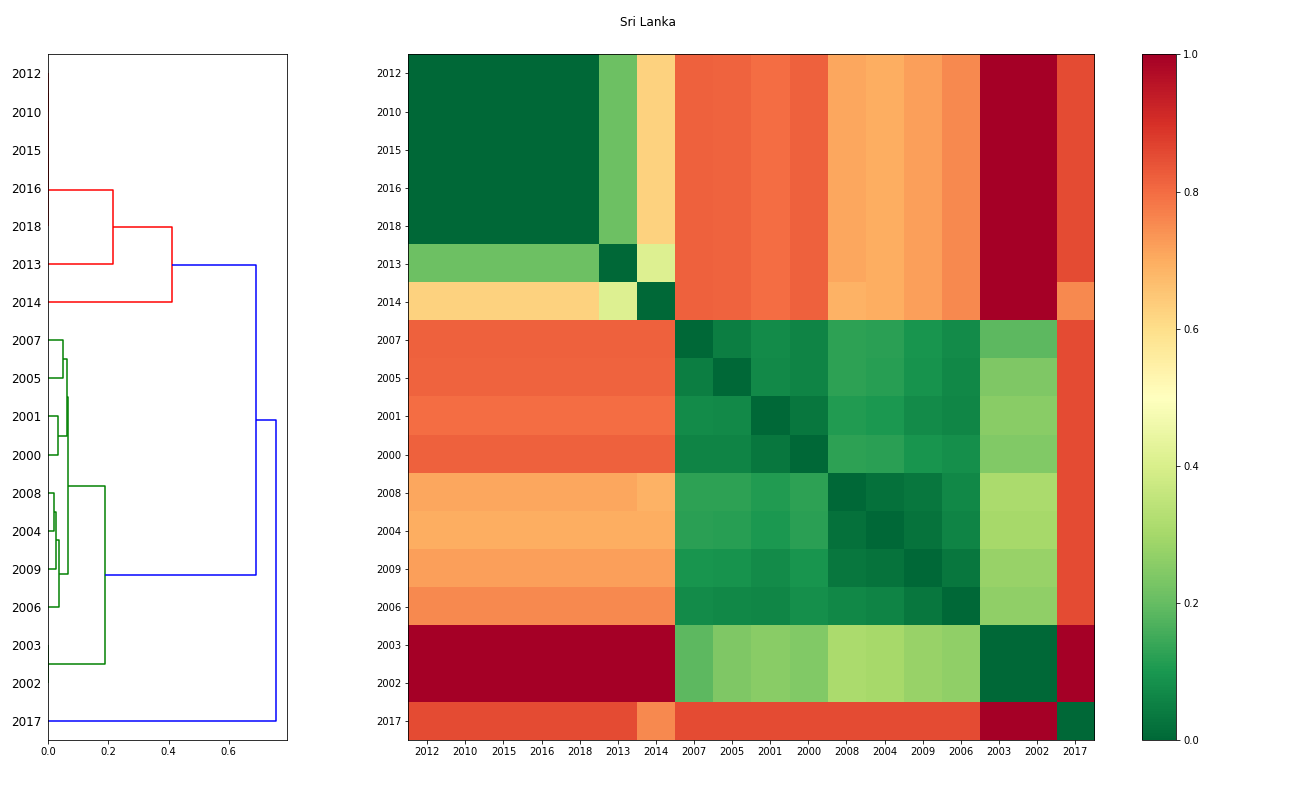}
        \caption{}
    \label{fig:SriLanka}
    \end{subfigure}
    \begin{subfigure}[b]{0.49\textwidth}
        \includegraphics[width=\textwidth]{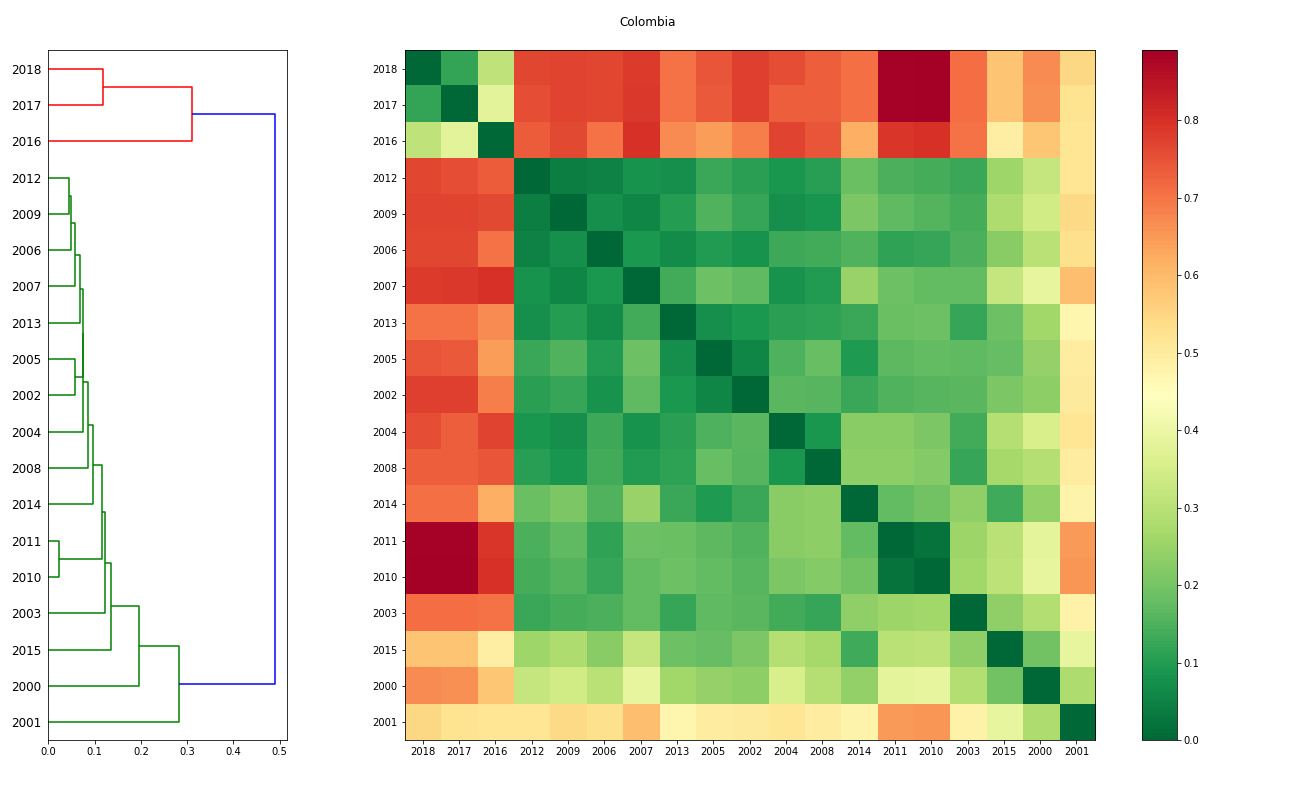}
        \caption{}
    \label{fig:Colombia}
    \end{subfigure}
    \begin{subfigure}[b]{0.49\textwidth}
        \includegraphics[width=\textwidth]{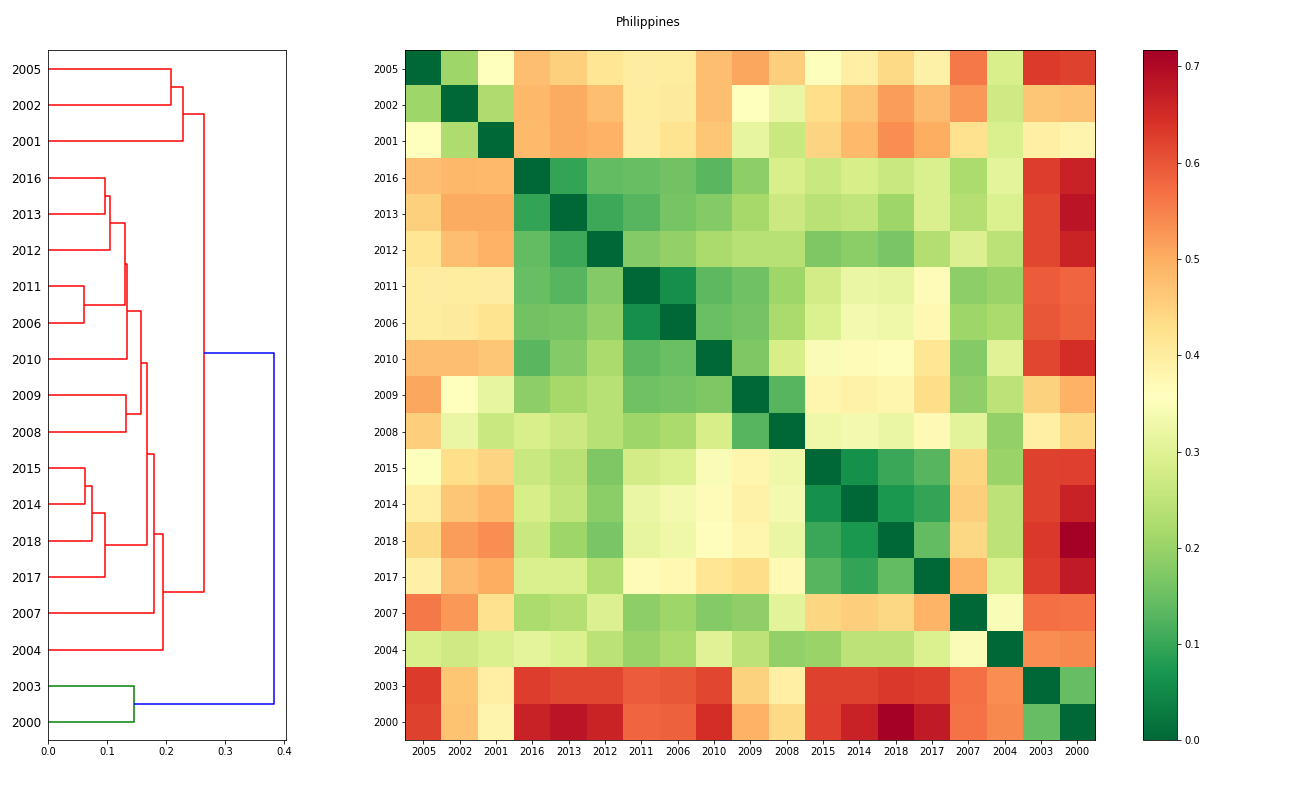}
        \caption{}
    \label{fig:Philippines}
    \end{subfigure}
    \caption{Hierarchical clustering on the perpetrator composition matrix $D^R$ for $R$ being one of six countries: (a) Afghanistan (b) the United States (c) Somalia (d) Sri Lanka (e) Colombia (f) the Philippines.}
    \label{fig:simplexplots}
\end{figure*}

\begin{figure*}
    \centering
    \begin{subfigure}[b]{0.85\textwidth}
        \includegraphics[width=\textwidth]{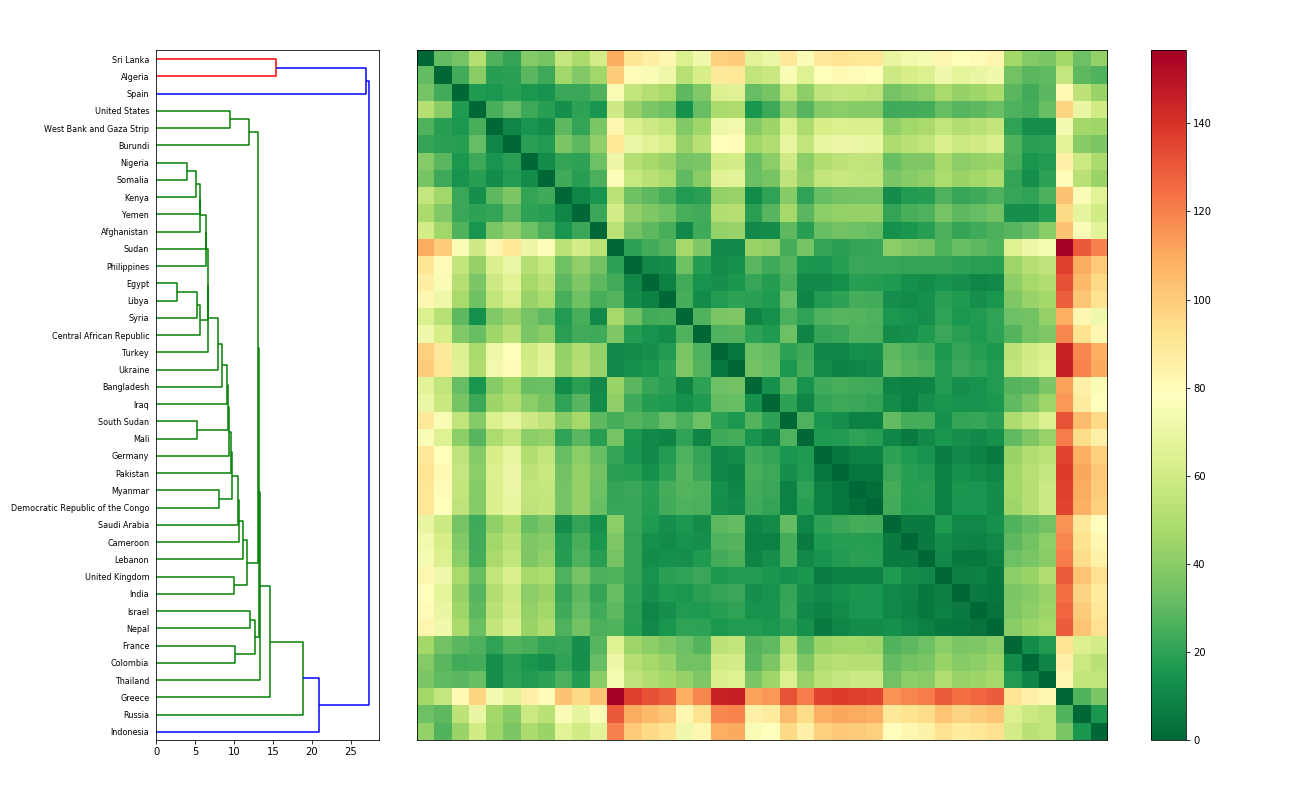}
        \caption{}
    \label{fig:Wassersteincountries}
    \end{subfigure}
    \begin{subfigure}[b]{0.85\textwidth}
        \includegraphics[width=\textwidth]{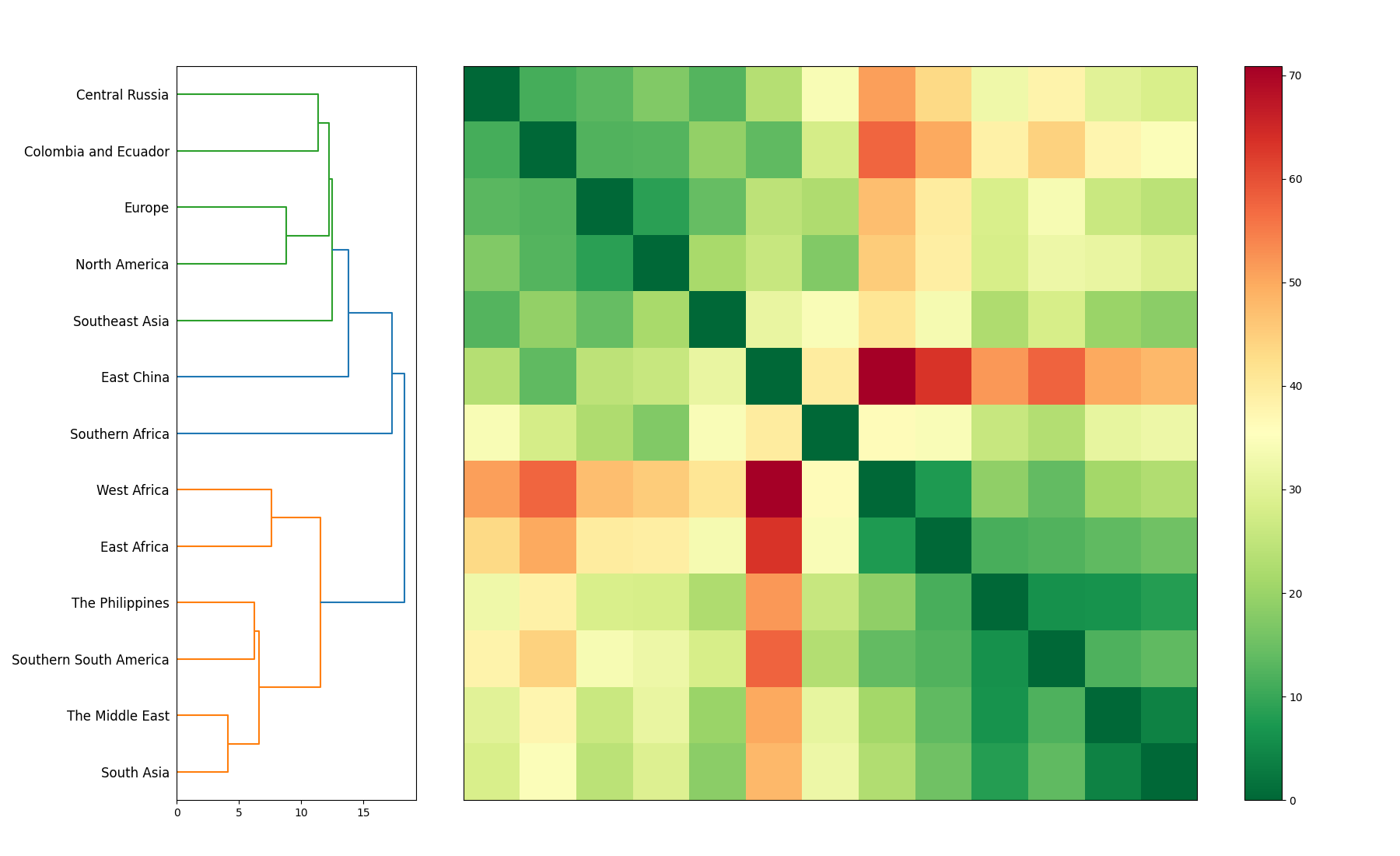}
        \caption{}
    \label{fig:Wassersteinclusters}
    \end{subfigure}
    \caption{Hierarchical clustering applied to the Wasserstein metric $D$ between time series distributions $f_R$ as $R$ ranges over our collection of (a) 40 countries and (b) 13 regions.}
    \label{fig:TSAWasserstein}
\end{figure*}

\begin{figure*}
    \centering
    \begin{subfigure}[b]{0.49\textwidth}
        \includegraphics[width=\textwidth]{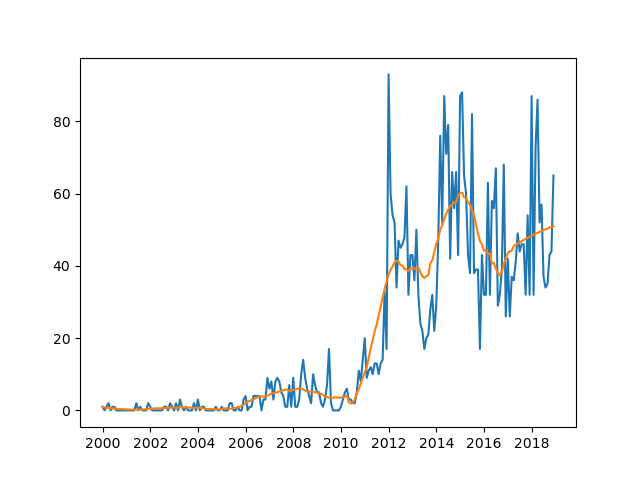}
        \caption{}
    \label{fig:NigeriaTS}
    \end{subfigure}
    \begin{subfigure}[b]{0.49\textwidth}
        \includegraphics[width=\textwidth]{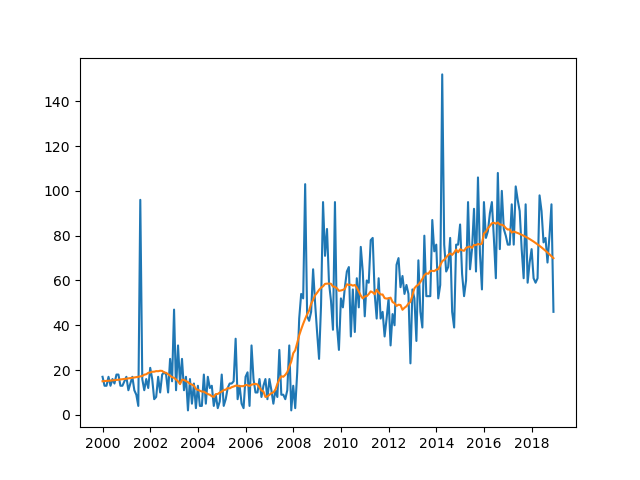}
        \caption{}
    \label{fig:IndiaTS}
    \end{subfigure}
    \begin{subfigure}[b]{0.49\textwidth}
        \includegraphics[width=\textwidth]{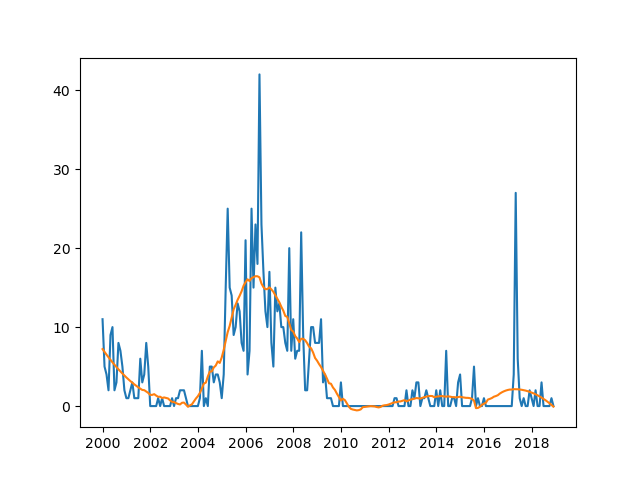}
        \caption{}
    \label{fig:SriLankaTS}
    \end{subfigure}
    \begin{subfigure}[b]{0.49\textwidth}
        \includegraphics[width=\textwidth]{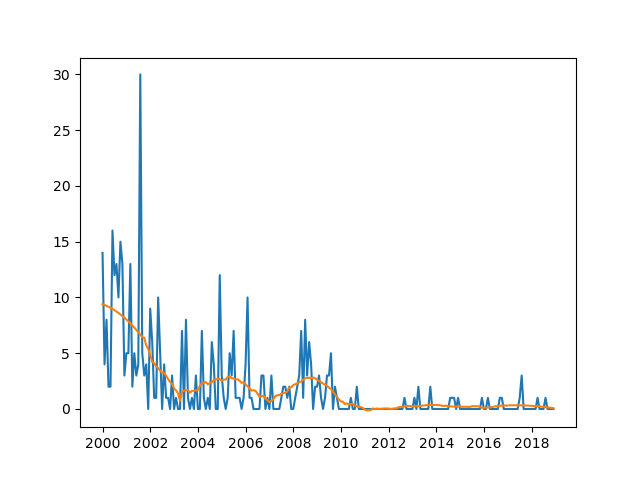}
        \caption{}
    \label{fig:SpainTS}
    \end{subfigure}
    \begin{subfigure}[b]{0.49\textwidth}
        \includegraphics[width=\textwidth]{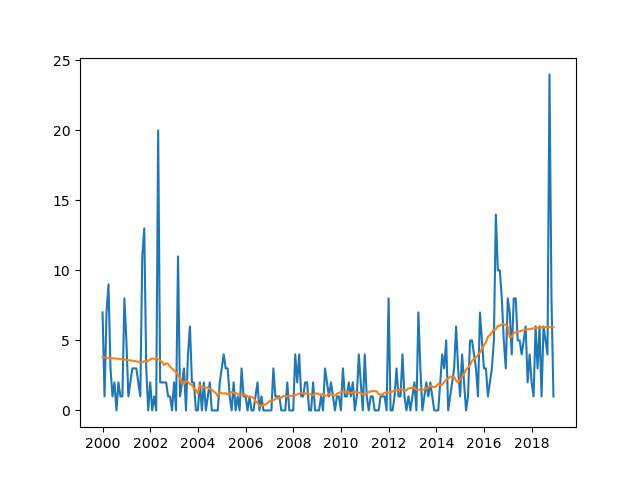}
        \caption{}
    \label{fig:USATS}
    \end{subfigure}
    \begin{subfigure}[b]{0.49\textwidth}
        \includegraphics[width=\textwidth]{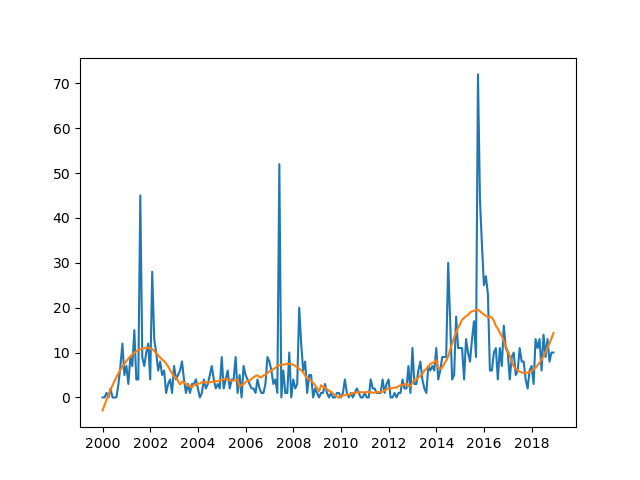}
        \caption{}
    \label{fig:GazaTS}
    \end{subfigure}
    \caption{Time series (and smoothing curve) of monthly terrorist events for (a) Nigeria (b) India (c) Sri Lanka (d) Spain (e) the United States (f) the West Bank and Gaza Strip.}
    \label{fig:TScountryplots}
\end{figure*}

%PYTHON indexing
% 0 - North America.
% 2 - Colombia, Ecuador
% 1 - south South America
% 3 (outlier, remove this) far east South America

% 4 - southern Africa
% 10 - East Africa
% 13 - West Africa
% 5 - Europe
% 12 - most Middle East
% 7 - Indian subcontinent
% 8 - SE Asia
% 11 - Philippines
% 14 - East China
% 9 (outlier, remove) far East Russia
% 6 (borderline outlier but enough to keep) middle of Russia

\section{Conclusion}
\label{sec:conclusion}

This paper has applied numerous mathematical approaches to analyze worldwide terrorism, drawing from metric geometry, unsupervised learning, and analysis of distributions and time series. We both reveal persistence and substantial changes in the global profile of terrorism. In Section \ref{sec:clusteranalysis}, we begin with a time-varying cluster analysis of terrorism on a geographic basis across the world, that is, where is it occurring. Cluster analysis is an unsupervised learning technique that generally intends to reveal structure in and simplify large datasets. The results are not always clearly or easily interpretable. In our case, we had positive outcomes: we showed substantial persistence in the cluster structure of terrorism on a geographic basis over time, and determined a convenient and interpretable list of 13 regions into which to divide all the world's terrorist events, other than a negligible proportion of outliers. These determined regions have then been used in other analysis in the paper, but could be used in concert with countless other methods in the applied mathematics and statistics literature, some of which have been applied to terrorism, but many of which have not. Alternatively, one could apply our cluster analysis on a more granular level, such as within a single country, where cluster identification could reveal structure in a setting where geographic locations may not be conducive to easy labeling. We discuss alternative future applications in the final paragraph of this section.

We must remark that our use of the word ``cluster'' is informed by the machine learning/statistical literature, rather than the standard security or foreign policy literature. In the latter, a ``terrorist cluster'' commonly refers to a collection of attacks or perpetrators on a very small geographic level, or a network of connected actors. The former would yield far too many clusters on the scale of the entire world to be informative, while the latter is not the focus of this paper (and is not summarized in the Global Terrorism Database).

In Section \ref{sec:simplexdistances}, we reveal that despite geographic persistence, there are ample changes in the composition of terrorist actors over time. And yet this is far from uniform among different countries and regions. We discuss geometric interpretations of distributions over perpetrator groups and settle on an appropriate discrepancy measurement with both geometric and probabilistic interpretations. Our results show substantial differences in the heterogeneity of changing perpetrator composition over time. In particular, regions and countries of Africa show the greatest heterogeneity, with African regions and countries represented among the highest ranks of heterogeneity in both Table \ref{tab:countrynormranking} and \ref{tab:clusternormranking}. We also discuss numerous high-profile examples of changing composition of terrorist activity, including the Tamil Tigers in Sri Lanka and the FARC in Colombia. Such insights are revealed to even authors with little specialized knowledge in terrorist activity, so we are confident our methodology may be used by experts in other fields to discern distributional changes both of a more subtle nature and in various other contexts.

Finally, Section \ref{sec:TSdistributions} judiciously selects an appropriate metric to understand the discrepancy between the noisy time series of monthly terrorist attack counts. Applying this to both countries and our determined regions, hierarchical clustering is able to quickly reveal noteworthy similarities and differences between distributions over time. The majority of countries, particularly from Africa, are revealed to have terrorist activity dominated by the last decade, but that is not uniform, with a more consistent pattern from the United States and West Bank/Gaza, and the opposite pattern from countries such as Algeria, Sri Lanka and Spain. Applying the same analysis to the determined regions reveals these geographic differences even more clearly. This section, having found broad differences in terrorist activity counts over time between continents and large regions of the world, complements Section \ref{sec:simplexdistances} well, particularly Table \ref{tab:clusternormranking}, where broad differences in the same regions are found with respect to the heterogeneity of perpetrator groups.

The primary limitation of this paper is its global scope. Future work could use these methods, including with more reliable and focused data sets, on a much more granular level, investigating, for example, terrorist activity within a single country over time. With greater precision in location reporting, geographic clusters of activity within a particular country could conceivably be far less consistent over time. One could repeat our analysis of compositional changes by removing ``unknown'' as the affiliation, or could aim to use investigative or supervised learning techniques to identify or guess the affiliation as much as possible. Further, one could incorporate relationships between different groups to modify the distance between distributions over perpetrators, such that related groups are considered closer than unrelated groups. And the number of attacks by different groups could be incorporated into the time series analysis of Section \ref{sec:TSdistributions}, including the use of parametric methods such as ARIMA models or frequency analysis. All our methods could be used in conjunction with a careful comparison of policy responses to identify outlier countries where the composition of terror perpetrators or the number of events is notably greater or lesser. This could reveal which policies were the most and least successful at a variety of goals, either the overall reduction of terrorism or the disruption and management of certain groups.

Overall, we believe this paper has added several mathematical approaches to the as-of-yet-limited use of mathematics, statistics and machine learning to analyze persistence and trends in terrorism, revealed elegant insights therein, and could be incorporated with the work of both other mathematicians and security experts to glean further insights into worldwide terrorism. Given the vast expenditures spent every year on combating terrorism, we would welcome further approaches to its research.

%-----------   END OF PAPER ---------------
% \begin{acknowledgements}
% Thanks to...
% \end{acknowledgements}
\section*{Conflict of interest}
The authors declare that they have no conflict of interest.

\section*{Data availability statement}
All data analyzed in this study are available at the Global Terrorism Database \cite{GTD}.

\appendix
\section{Existing cluster theory}
\label{app:cluster}

In this section, we provide an overview of the three clustering algorithms required in the above implementations: hierarchical clustering, K-means and K-medoids, and methods to select the optimal number of clusters. In our most general setup, $x_1,\dots,x_n$ are elements of some normed or metric space $\mathfrak{X}$.

\emph{Hierarchical clustering} \cite{Ward1963} is an iterative clustering technique that does not specify discrete partitions of elements. Rather, the methodology seeks to build a hierarchy of similarity between elements. In this paper we use agglomerative hierarchical clustering with the average linkage method \cite{Mllner2013}, where each element $x_i$ begins in its own cluster and branches between them are successively built. The results of hierarchical clustering are commonly displayed in \emph{dendrograms}, as we use in this paper. Unlike K-means, hierarchical clustering does not require the choice of a number of clusters $k$. Further, hierarchical clustering can be implemented on any distance, not necessarily Euclidean space.

\emph{K-means and K-medoids clustering} seek to minimize appropriate sums of square distances. With $k$ chosen \textit{a priori}, we investigate all possible partitions (disjoint unions)  $C_1 \cup C_2 \cup \dots \cup C_k $ of $\{ x_1,\dots, x_n \}$. Let $z_j$ be the \emph{centroid} (average) of the subset $C_j$. For K-means, one seeks to minimize the sum of square distances within each cluster to its centroid:
$$ \sum_{j=1}^k \sum_{x \in C_j} ||x - z_j||^2 $$
For a normed space with dimension of at least 2, it is NP-hard to find the global optimum of this problem. The K-means algorithm \cite{Lloyd1982} is an iterative algorithm that converges quickly and suitably to a locally optimal solution. K-means can typically only be implemented on Euclidean space. K-medoids is a substitute where $z_j$ is restricted to be one of the data points, and selected such that its distance from other elements of the cluster is minimal. Thus, K-medoids can be implemented more generally on any metric space, not necessarily on a normed space with the Euclidean metric.

How to best choose the number of clusters $k$ for the K-means or K-medoids algorithm is a difficult problem. Different methods for estimating $k$ may produce considerably differing results. In this paper, we draw upon 16 methods to determine the appropriate number of clusters before implementing K-means or K-medoids. These methods include among others, Ptbiserial index \cite{Milligan1980}, silhouette score \cite{Rousseeuw1987}, KL index \cite{krzanowski1988}, C index \cite{Hubert1976}, McClain-Rao index \cite{Mcclain1975} and Dunn index \cite{Dunn1974}. The methodology described above is flexible however, and may use any combination of existing methods.

\section{Probability distribution distance}
\label{app:discreteWass}

Let $(X,d)$ be any metric space, $\mu,\nu$ two probability measures on $X$, and $q \geq 1$. The Wasserstein metric between $\mu, \nu$ is defined as
\begin{align}
\label{eq:Wasserstein}
    W^q (\mu,\nu) = \inf_{\gamma} \bigg( \int_{X \times X} d(x,y)^q d\gamma  \bigg)^{\frac{1}{q}},
\end{align}
where the infimum is taken over all probability measures $\gamma$ on $X \times X$ with marginal distributions $\mu$ and $\nu$. Henceforth, let $q=1$. By the Kantorovich-Rubinstein formula \cite{Kantorovich}, there is an alternative formulation when X is compact (for example, finite):
\begin{align}
\label{eq:Wassersteinalt}
    W^1 (\mu,\nu) = \sup_{F} \left| \int_{X} F d\mu - \int_{X} F d\nu  \right|,
\end{align}
where the supremum is taken over all $1$-Lipschitz functions $F: X \to \mathbb{R}$.

\begin{proposition}
Let $(X,d)$ be a finite discrete metric space, with $d(x,y)=1$ for all $x\neq y$ and 0 otherwise. Let $W^1(\mu,\nu)$ be the $L^1$-Wasserstein metric between two probability measures $\mu,\nu$ on $X$, with corresponding distribution functions $f,g$, expressed as in (\ref{eq:Wassersteinalt}). That is,
\begin{align}
\label{eq:Wassersteinalt_new}
    W^1 (\mu,\nu) = \sup_{F} \left| \int_{X} F d\mu - \int_{X} F d\nu  \right|.
\end{align}
Then, this supremum is optimized by the following choice of $F$:
\begin{align}
\label{eq:bestF_new}
F(x) = \begin{cases}
1, & f(x) \geq g(x), \\
0, & f(x)<g(x).
\end{cases}
\end{align}
Thus, $W^1(f,g)$ reduces to the same form of (\ref{eq:simplexdist}), namely
\begin{align}
\label{eq:discreteWass_new}
W^1(f,g)=\frac12 ||f - g||_1.
\end{align}
\end{proposition}
\begin{proof}
Let $F$ be an arbitrary $1$-Lipschitz function on $X$. That is, $F: X \to \mathbb{R}$ and $|F(x) - F(y)| \leq d(x,y) \leq 1$ for all $x,y$. Let $M=\sup_{x\in X} F(x)$ and $ m= \inf_{y \in X} F(y).$ By taking the supremum over $x$ and the infimum over $y$, the Lipschitz condition implies that $M-m\leq 1.$ So
\begin{align}
  \int_{X} F d\mu - \int_{X} F d\nu  &=\sum_{x \in X} F(x)(f(x)-g(x)) \\
  &\leq \sum_{x: f(x) \geq g(x) } M(f(x)-g(x)) + \sum_{x: f(x) < g(x) } m(f(x)-g(x))\\
  &\leq \sum_{x: f(x) \geq g(x) } (m+1)(f(x)-g(x)) + \sum_{x: f(x) < g(x) } m(f(x)-g(x))\\
  &=\sum_{x: f(x) \geq g(x) } f(x)-g(x) + \sum_{x \in X}m(f(x)-g(x))\\
  &=\sum_{x: f(x) \geq g(x) } f(x)-g(x),
\end{align}
using the fact that $\sum_x f(x)=\sum_x g(x)=1$ to eliminate the second summand. Next, let
\begin{align}
    P&=\sum_{x: f(x) \geq g(x) } f(x)-g(x),\\
    N&=\sum_{x: f(x) < g(x) } g(x)-f(x).
\end{align}
So $P-N = \sum_{x\in X} f(x)-g(x) = 0,$ while $P+N = \sum_{x \in X} |f(x)-g(x)|=\|f-g\|_1.$ Thus, we see $P=N=\frac12 \|f-g\|_1,$ and $\left|\int_{X} F d\mu - \int_{X} F d\nu\right| \leq P$. Taking the supremum over $F$, we derive $W^1(f,g) \leq P$. Finally, let $F$ be the function defined in (\ref{eq:bestF_new}). Then $\int_{X} F d\mu - \int_{X} F d\nu$ coincides with $P$ by definition. Thus, the supremal value coincides exactly with $P$, and $P=\frac12 \|f-g\|_1$, as required.

\end{proof}

\bibliographystyle{_elsarticle-num-names}
\bibliography{__newrefs} 
\biboptions{sort&compress}
\end{document}